\documentclass[11pt,twoside]{amsart}
\usepackage{latexsym,amssymb,amsmath}
\usepackage[all]{xy}
\usepackage{enumerate}
\usepackage{extpfeil}
\usepackage{afterpage}
\usepackage{float}
\usepackage{hyperref}
\usepackage{amsmath, amsthm, latexsym, amssymb, amsfonts, epsf, xcolor, hyperref}
\usepackage{tikz,pgfplots}
\usepackage[inline]{enumitem}

\numberwithin{equation}{section}
\newtheorem{theorem}{Theorem}[section]
\newtheorem{lemma}[theorem]{Lemma}
\newtheorem{proposition}[theorem]{Proposition}
\newtheorem{corollary}[theorem]{Corollary}

\theoremstyle{definition}
\newtheorem{definition}[theorem]{Definition}
\newtheorem{procedure}[theorem]{Procedure}
\newtheorem{remark}[theorem]{Remark}
\newtheorem{example}[theorem]{Example}

\newcommand{\hir}[1]{\textcolor{red}{#1}}
\newcommand{\cb}{\color {blue}}
\newcommand{\rmv}[1]{}
\begin{document}


\title[Relative generalized Hamming weights of evaluation codes]
{Relative generalized Hamming weights of evaluation codes}

\author[D. Jaramillo-Velez]{Delio Jaramillo-Velez}
\address{(Delio Jaramillo-Velez) 
Departamento de
Matem\'aticas\\
Centro de Investigaci\'on y de Estudios Avanzados del IPN\\
Apartado Postal
14--740 \\
07000 Mexico City, Mexico.
}
\email{djaramillo@math.cinvestav.mx}

\author[H. H. L\'opez]{Hiram H. L\'opez}
\address{(Hiram H. L\'opez) Department of Mathematics and Statistics, Cleveland State University, Cleveland, OH USA}
\email{h.lopezvaldez@csuohio.edu}

\author[Y. Pitones]{Yuriko Pitones}
\address{(Yuriko Pitones) Departmento de Matem\'aticas, Universidad Aut\'onoma Metropolitana-Iztapalapa, Ciudad de M\'exico, M\'exico.
} 
\email{ypitones@xanum.uam.mx}

\thanks{The first author was partially supported by a CONACYT scholarship. The second author was partially supported by an AMS--Simons Travel Grant. The third author was partially supported by CONACyT and SNI}
\begin{abstract} 
The aim of this work is to algebraically describe the relative generalized Hamming weights of evaluation codes.
We give a lower bound for these weights in terms of a footprint bound. We prove that this bound can be sharp.
We compute the next-to-minimal weight of toric codes over hypersimplices of degree 1.
\end{abstract}

\dedicatory{Dedicated to Rafael H. Villarreal on the occasion of his seventieth birthday.}

\maketitle

\section{Introduction}\label{intro-section}
Let $K:=\mathbb{F}_q$ be a finite field with $q$ elements,
$S:=K[t_1,\ldots,t_s]=\bigoplus_{d=0}^\infty S_d$ a polynomial ring over $K$
with the standard grading, and 
$X:=\{P_1,\ldots,P_m\}$ a set of distinct points in the affine space
$\mathbb{A}^s:=K^s$. The \textit{evaluation map} 
is the $K$-linear map given by
$${\rm ev}\colon S\rightarrow K^{m},\quad
f\mapsto\left(f(P_1),\ldots,f(P_m)\right).
$$ 
\quad The kernel of ${\rm ev}$, denoted by $I(X)$, is the 
\textit{vanishing ideal} of $X$ consisting of the polynomials of $S$
that vanish at all points of $X$. This map induces an isomorphism of
$K$-linear spaces between $S/I(X)$ and $K^m$. Let $\mathcal{L}$ be a
linear subspace of $S$ of finite dimension. The image of $\mathcal{L}$
under the evaluation map, denoted by $\mathcal{L}_X$, is called an
\textit{evaluation code} on $X$ \cite{stichtenoth,tsfasman}. 
Fix an integer $d\geq 1$ and let $S_{\leq
d}:=\bigoplus_{i=0}^dS_i$ be the
$K$-linear subspace of $S$ of all polynomials of degree at most $d$. If
$\mathcal{L}$ is equal to $S_{\leq d}$, then the resulting evaluation
code $\mathcal{L}_X,$ denoted by $C_X(d),$ is called an 
\textit{affine Reed-Muller-type code} of degree $d$ on $X$
\cite{duursma-renteria-tapia,GRT}.

Let $\prec$ be a monomial order on $S$ and $I:=I(X)$ the
vanishing ideal of $X$. We
denote the {initial monomial} of a non-zero polynomial 
$f\in S$ by ${\rm in}_\prec(f),$ and the initial ideal of $I$ by 
${\rm in}_\prec(I)$. A monomial $t^a:=t_1^{a_1}\cdots t_s^{a_s}$,
$a=(a_1,\dots,a_s) \in \mathbb{Z}^s_{\geq 0}$, 
is called a 
\textit{standard monomial} of $S/I$, with respect 
to $\prec$, if $t^a\notin{\rm in}_\prec(I)$. 
The \textit{footprint} of $S/I$, denoted
$\Delta_\prec(I)$, is the set of all standard monomials of $S/I$. The
footprint has been extensively utilized for the study of evaluation codes 
\cite{datta,geil-2008,geil-hoholdt,geil-pellikaan,rth-footprint,Pellikaan,Jaramillo-Vaz-pinto-Villarreal,cartesian-codes}.

If $\mathcal{A}\subseteq S$, the linear subspace of $S$ spanned by $\mathcal{A}$ is denoted by $K\mathcal{A}$. The linear code $\mathcal{L}_X$ is called a \textit{standard
evaluation code} on $X$ relative to $\prec$ if 
$\mathcal{L}$ is a linear subspace 
of $K\Delta_\prec(I)$. 
A polynomial $f$ is called
a \textit{standard polynomial} of $S/I$ if $f\neq 0$ and $f$ is in
$K\Delta_\prec(I)$. 
As the field $K$ is finite, there are only a finite number of standard
polynomials. The standard evaluation codes are those evaluation codes 
defined by vector spaces generated by standard polynomials. The dual of an evaluation code $\mathcal{L}_{X}$, denoted by $(\mathcal{L}_{X})^{\perp}$, is the set of all $\alpha\in K^{m}$ such that $\langle\alpha,\beta\rangle=0$ for all $\beta\in\mathcal{L}_{X}$, where $\langle\:,\:\rangle$ is the ordinary inner product in $K^{m}$. In \cite{Lopez-Soprunov-Villarreal}, the code $(\mathcal{L}_{X})^{\perp}$ is studied in terms of  standard monomials and indicator functions. Among the main properties, it is shown that the dual of an evaluation code is an evaluation code.

Let $C \subseteq K^m$ be a linear code. The {\it Hamming weight}, or just {\it weight}, of an element $c \in C$ is the number of nonzero coordinates in $c.$
Let $A_i (C)$ be the number of codewords of weight $i$ in $C.$ The list $A_i$ for $0 \leq i \leq m$ is called the {\it weight distribution} of $C.$  The elements $i>0$ for which $A_i (C)\neq 0$ are called the {\it Hamming weights} or {\it next-to-minimal weights} of $C.$ Thus, the {\it minimum distance} of $C$ is the first non-zero Hamming weight. The Hamming weights allow us the computation of important efficiency parameters for the code $C$, such as the Probability of Undetected Error {\rm \cite{WML82}}. 

Cherdieu, Rolland, Geil, Bruen, and Leduc studied the weight distribution of generalized Reed-Muller codes in {\rm \cite{Bruen,Cherdieu-Rolland,geil,Leduc,Rolland}}. More recently, Carvalho and Neumann have determined the next-to-minimal weights of  affine Cartesian codes,  binary projective Reed-Muller codes, and projective Reed-Muller codes {\rm \cite{Carvalho,Car-Neu-pro-RM2,Car-Neu-af2,Car-Neu-pro-RM1,Car-Neu-af1,Car-Neu-bi}}. In this work, we study the next-to-minimal weights of certain evaluation codes called toric codes over hypersimplices.

The relative generalized Hamming weights (RGHW) of two linear codes $C_2 \subsetneq C_1$ were introduced by Luo et al. in \cite{Luo-Mitrpant-Vinck-Chen}. In the same work, the relative dimension/length profile is also extended. The RGHW are a natural generalization of the generalized Hamming weights~\cite{Wei}. The RGHW are important to design a perfect secrecy coding scheme for the coordinated multiparty model. As another application, the $m$-th RGHW of two codes $C_2 \subsetneq C_1$ expresses the smallest size of unauthorized sets that can obtain $m$ $q$-bits {\rm \cite{Bains}}, where $q$ is the size of the alphabet of the codes $C_1$ and $C_2$.

The RGHW of certain important families of linear codes have been studied recently.
Some properties of the RGHW of one-point algebraic geometry codes are described in {\rm\cite{GMMRL}}.
O. Geil and M. Stefano obtained the relative generalized Hamming weights of $q$-ary Reed-Muller codes using the footprint bound \cite{geil-stefano}. M. Datta followed the footprint bound technique and presented a combinatorial formula for the relative generalized Hamming weights of affine Cartesian codes {\rm \cite{datta}}. These works introduced methods to study certain evaluation codes. In this work, we algebraically describe the RGHW of evaluation codes in terms of the degree of $I(X).$

The contents of this paper are as follows. In Section \ref{Prelim-section}, we introduce some concepts that will be needed throughout the paper. 
In Section \ref{RGHW-section}, we show our main results related to RGHW. Theorem \ref{relative-Hamming-weights}  gives an algebraic equivalence  for the relative generalized Hamming weights of evaluation codes. In Section \ref{sec-weight-distribution}, we present  the next-to-minimal weights of toric codes over hypersimplices of degree one. We also provide bounds for the maximum number of zeros of homogeneous squarefree polynomials over the affine torus. Section \ref{example-sec} includes implementations in Macaulay2 \cite{mac2} and examples that ilustrare some of our results.

For additional information about Gr\"obner bases and Commutative Algebra, we refer to {\rm \cite{monalg-rev}}. For basic Coding Theory, we refer to {\rm \cite{MacWilliams-Sloane}}. 


\section{Preliminaries}\label{Prelim-section} 
In this section, we explore some well-known definitions and results that we use in the following sections. In particular, we revisit the notion of affine Hilbert functions and the degree of an ideal, which help us to determine the maximum number of common zeros of a system of polynomials. We also recall the definition of relative generalized Hamming weight of a code. For a more detailed description of the results that are mentioned here we refer to {\rm \cite{CLO,Luo-Mitrpant-Vinck-Chen,MacWilliams-Sloane,monalg-rev}}.

Let $I$ be an ideal of $S=K[t_1,\ldots,t_s].$ The Krull dimension of $S/I$ is denoted by $\dim(S/I)$. We say that the ideal $I$ has {\it dimension} $k$ if $\dim(S/I)$ is equal to $k$. The \textit{height} of $I$ is denoted and defined by ${\rm ht}(I):=s-\dim(S/I)$. The $K$-linear space of polynomials in $S$  (resp. $I$) of degree at most $d$ is denoted by $S_{\leq d}$ (resp. $I_{\leq d}$). The function
$$
H_I^a(d):=\dim_K(S_{\leq d}/I_{\leq d}),\ \ \ d=0,1,2,\ldots,
$$
is  called the \textit{affine Hilbert function} of $S/I$. 
Let $u=t_{s+1}$ be a new variable and let $I^h\subseteq S[u]$  be the {\it homogenization\/} of $I$, where $S[u]$ is given the  standard grading.  One has the following two well-known facts
$$
\dim(S[u]/I^h)=\dim(S/I)+1\mbox{ and } H_I^a(d)=H_{I^h}(d)
\mbox{ for }d\geq 0,
$$
where $H_{I^h}(d)=\dim_K(S[u]/I^h)_d$. See for instance \cite[Lemma~8.5.4]{monalg-rev}. By a Hilbert theorem \cite[p.~58]{Sta1}, there is a unique polynomial $h^a_I(z)=\sum_{i=0}^{k}a_iz^i\in  \mathbb{Q}[z]$ of degree $k$ such that $h^a_I(d)=H_I^a(d)$ for $d\gg 0$. By convention the degree of the zero polynomial is $-1$. The integer $k!\, a_k$, denoted ${\rm deg}(S/I)$, is called the \textit{degree} of $S/I$.  The degree of $S/I$ is equal to  $\deg(S[u]/I^{h})$. If $k=0$, then $H_I^a(d)=\deg(S/I)=\dim_K(S/I)$ for $d\gg 0$. Note that the degree of $S/I$ is positive if $ I\subsetneq S$ and is $0$ otherwise. 

Let $F$ be a finite subset of $S$ and $X=\{P_1,\ldots,P_m\}$ a set of distinct points in the affine space
$\mathbb{A}^s$. The \textit{affine variety} of $F$ in $X$, denoted by $V_X(F)$, is the set of all $P \in X$ such that $f(P)=0$ for all $f\in F$.  The {\it colon ideal} is defined by  $$(I\colon(F)):=\{g\in S\, \vert\, gF\subseteq I \}.$$ The colon idea is an important tool to determine whether or not the affine variety $V_X(F)$ is non-empty (Lemma~\ref{vila-delio-feb27-20}).

We now recall some important properties.

\begin{lemma}{\cite[p.~389]{cocoa-book}}\label{primdec-ixx} 
Let $X$ be a finite subset of
$\mathbb{A}^s$, let $P$ be a point in $X$, $P=(p_1,\ldots,p_s)$, and
let $I_{P}$ be the vanishing ideal 
of $P$. Then $I_P$ is a maximal ideal of height $s$, 
\begin{equation*}
I_P=(t_1-p_1,\ldots,t_s-p_s),\ \deg(S/I_P)=1, 
\end{equation*}
and $I(X)=\bigcap_{P\in X}I_{P}$ is the primary 
decomposition of $I(X)$.  
\end{lemma}

\begin{lemma}{\rm(cf. \cite[Lemma~3.3]{Bernal-Pitones-Villarreal})}\label{apr5-19}
Let $d_1,\ldots,d_s$ be positive integers and let $L$ be the ideal of $S$ generated by
$t_1^{d_1},\ldots,t_{s}^{d_s}$. If $t^a=t_1^{a_1}\cdots t_s^{a_s}$ is
not in $L$, then
$$
\deg(S/(L,t^a))=d_1\cdots d_s-(d_1-a_1)\cdots(d_s-a_{s}).
$$
\end{lemma}

\begin{lemma}{\rm \cite[Lemma 2.5]{Jaramillo-Vaz-pinto-Villarreal}}\label{vila-delio-feb27-20} Let $X$ be a finite subset of 
$\mathbb{A}^s$ over a field $K$ and let $F=\{f_1,\ldots,f_r\}$ be a
set of polynomials of $S$. Then, the following conditions are equivalent.
\begin{enumerate}
\item[(a)] $(I(X)\colon(F))=I(X)$.
\item[(b)] $V_X(F)=\emptyset$.
\item[(c)] $(I(X),F)=S$.
\end{enumerate}
\end{lemma}

Let $\prec$ be a monomial order and
$\Delta_\prec(I)$ the set of standard monomials of $S/I$. 
The image of $\Delta_\prec(I)$, under the canonical 
map $S\mapsto S/I$, $x\mapsto \overline{x}$, is a basis of $S/I$ as a
$K$-vector space \cite[Proposition~6.52]{Becker-Weispfenning}. In particular, $H_{I}^{a}(d)$ is the number of standard monomials of $S/I$ of degree at most $d$. For the set of polynomials $F\subseteq S$ we set ${\rm in}_{\prec}(F):=\lbrace {\rm in}_{\prec}(f): f\in F\rbrace$. 
A subset $\mathcal{G}=\{g_1,\ldots, g_n\}$ of $I$ is called a 
{\it Gr\"obner basis\/} of $I$ if ${\rm
in}_\prec(I)=({\rm in}_\prec(g_1),\ldots,{\rm in}_\prec(g_n))$.

\begin{theorem}{\rm \cite[Theorem 2.12]{Jaramillo-Vaz-pinto-Villarreal}}\label{degree-initial-theorem}
Let $X$ be a finite subset of $\mathbb{A}^{s}$, let $I=I(X)$ be the vanishing ideal of $X$, and let $\prec$ be a monomial order. If $F$ is a finite set of polynomial of $S$ and $(I:(F))\neq I$, then
$$
|V_{X}(F)|={\rm deg}(S/\left( I,F\right) )\leq{\rm deg}(S/\left( {\rm in}_{\prec}(I),{\rm in}_{\prec}(F)\right) )\leq {\rm deg}(S/I )=|X|,
$$
and ${\rm deg}(S/\left( I,F\right) )<{\rm deg}(S/I )$ if $(F)\not\subseteq I$.
\end{theorem}

We now recall the definition of the relative generalized Hamming weights of a linear code with respect to a proper subcode.

\begin{definition}\cite{Liu-Chen-Luo}
Let $C_{2}\subsetneq C_{1}$ two linear codes. For $r\in\lbrace1,\dots,{\rm dim}(C_{1})-{\rm dim}(C_{2})\rbrace$, the $r$-th relative generalized Hamming weight of $C_{1}$ with respect to $C_{2}$ is denoted and defined as 
$$
M_{r}(C_{1},C_{2}):=\min_{J\subseteq[n]}\lbrace |J|:{\rm dim}((C_{1})_{J})-{\rm dim}((C_{2})_{J})=r\rbrace,
$$
where $(C_{i})_{J}=\lbrace (c_{1},\dots,c_{n})\in C_{i}| c_{t}=0\text{ for } t\not\in J\rbrace$ for $i=1,2$.
\end{definition}

The {\it support} $\mathcal{X}(D)$ of a code $D \subseteq K^m$ is defined by the set 
$$
\mathcal{X}(D):=\lbrace i \mid \exists(a_{1},\dots,a_{m})\in D, a_{i}\neq 0\rbrace.
$$
The {\it support} $\mathcal{X}(\alpha)$ of a vector $\alpha\in K^{m}$ is the set of all non-zero entries of $\alpha$.

The following lemma gives us an alternative description for the relative generalized Hamming weights

\begin{lemma}\cite[Lemma 1]{Liu-Chen-Luo}\label{RGHW-description}
Let $C_{2}\subsetneq C_{1}$ be linear codes. For $r=1,\dots,{\rm dim}(C_{1})-{\rm dim}(C_{2})$, we have 
\begin{equation}\label{21.11.13}
M_{r}(C_{1},C_{2}):=\min\lbrace |\mathcal{X}(D)|: D\subseteq C_{1}; D\cap C_{2}=\lbrace0\rbrace,{\rm dim}(D)=r\rbrace,
\end{equation}
where $D\subseteq C_{1}$ represents that $D$ is a subcode of $C_1.$
\end{lemma}

\begin{remark}\label{case-of-gener-Ha-we}
Note that if $C_{2}=\lbrace0\rbrace$, then the $r$-th relative generalized Hamming weight of $C_{1}$ with respect to $C_{2}$ is equal to the $r$-th generalized Hamming weight of $C_{1}$. 
\end{remark}


\section{An algebraic representation}\label{RGHW-section} 
In this section, we algebraically describe the relative generalized Hamming weights of evaluation codes $\mathcal{L}_{X}^{2} \subsetneq \mathcal{L}_{X}^{1}$ in terms of the degree of $I=I(X).$ In addition, we provide a lower bound for these relative weights in terms of the footprint.

Given two linear subspaces $\mathcal{L}^{2}\subsetneq\mathcal{L}^{1}$ of $K\Delta_{\prec}(I)$, we denote by $\mathcal{L}^{1}/\mathcal{L}^{2}$ the quotient linear space of $\mathcal{L}^{1}$ and $\mathcal{L}^{2}$, and by $\mathcal{L}^{1}\setminus\mathcal{L}^{2}$ the set of elements of $\mathcal{L}^{1}$ that are not elements in $\mathcal{L}^{2}$. We denote by $\mathcal{R}_{r}$ as the set of all subsets $F=\lbrace f_{1},\dots,f_{r}\rbrace \subseteq \mathcal{L}^{1}$ of size $r$ such that the elements $\overline{f_{1}},\dots,\overline{f_{r}} \in \mathcal{L}^{1}/\mathcal{L}^{2}$ are linearly independent, where $\overline{f_{i}}:= f_i + \mathcal{L}^{2},$ for $i=1,\ldots, r.$
Let $\mathcal{R}_{\prec,r}$ be the set of all subsets $F=\lbrace f_{1},\dots,f_{r}\rbrace \subseteq \mathcal{L}^{1}\setminus\mathcal{L}^{2}$ of size $r$ such that ${\rm in}_{\prec}(f_{1}),\dots,{\rm in}_{\prec}(f_{r})$ are distinct monomials and $f_{i}$ is monic for $i=1,\dots,r$. 

\begin{lemma}\label{quotient-space-lemma}
Let $\mathcal{L}^{2}\subsetneq\mathcal{L}^{1}$ be two $K$-linear subspaces  of $K\Delta_{\prec}(I)\subseteq S$ of finite dimension, and let $F=\lbrace f_{1},\dots,f_{r}\rbrace$ be a subset of  $\mathcal{L}^{1}\setminus\lbrace0\rbrace$. Then, the following hold.
\begin{itemize}
\item[(a)] If $\overline{f_{1}},\dots,\overline{f_{r}}\subseteq\mathcal{L}^{1}/\mathcal{L}^{2}$ are linearly independent over $K$, then there is $G=\lbrace g_{1},\dots,g_{r}\rbrace$ $\subseteq \mathcal{L}^{1}\setminus\mathcal{L}^{2}$ such that $KF=KG$, ${\rm in}_{\prec}(g_{1}),\dots,{\rm in}_{\prec}(g_{r})$ are distinct, and ${\rm in}_{\prec}(f_{i})\succeq{\rm in}_{\prec}(g_{i})$ for all $i$.
\item[(b)] If ${\rm in}_{\prec}(f_{1}),\dots,{\rm in}_{\prec}(f_{r})$ are distinct, and $f_{1},\dots f_{r}\in\mathcal{L}^{1}\setminus\mathcal{L}^{2}$, then $\overline{f_{1}},\dots,\overline{f_{r}}$ are linearly independent over $K$.
\item[(c)] $\mathcal{R}_{\prec,r}\subseteq \mathcal{R}_{r}$, and if $F$ is in $\mathcal{R}_{r}$, then there is $G$ in $\mathcal{R}_{\prec,r}$ such that $KF=KG$.
\end{itemize}
\end{lemma}
\begin{proof}
(a) As $\overline{f_{1}},\dots,\overline{f_{r}}\subseteq\mathcal{L}^{1}/\mathcal{L}^{2}$ are linearly independent over $K$, then $f_{1}\dots,f_{r}$ are linearly independent over $K$. We proceed by induction on $r$. For the  case $r=1,$ take $g_1:=f_1.$ Assume that $r>1$. If it is needed, permute the $f_{i}$'s so we have that ${\rm in}_{\prec}(f_{1})\succeq\dots\succeq{\rm in}_{\prec}(f_{r})$.
\begin{itemize}
\item Assume ${\rm in}_{\prec}(f_{1})\succ{\rm in}_{\prec}(f_{2})$. By applying the induction hypothesis to the set $H=\lbrace f_{2},\dots,f_{r}\rbrace$, we obtain a set $G'=\lbrace g_{2}\dots,g_{r}\rbrace\subseteq \mathcal{L}^{1}\setminus\mathcal{L}^{2}$ such that $KH=KG'$, the monomials ${\rm in}_{\prec}(g_{2}),\dots,{\rm in}_{\prec}(g_{r})$ are distinct, and ${\rm in}_{\prec}(f_{i})\succeq{\rm in}_{\prec}(g_{i})$ for all $i\geq2$. Setting $g_{1} := f_{1}$ and $G := G'\cup\lbrace g_{1}\rbrace$, we get $KF=KG$, and the monomial ${\rm in}_{\prec}(g_{1})$ is distinct from ${\rm in}_{\prec}(g_{2}),\dots,{\rm in}_{\prec}(g_{r})$ because ${\rm in}_{\prec}(f_{1})\succ{\rm in}_{\prec}(f_{i})\succeq{\rm in}_{\prec}(g_{i})$ for $i\geq2.$
\item Assume that there is $k\geq2$ such that ${\rm in}_{\prec}(f_{1})={\rm in}_{\prec}(f_{i})$ for $i\leq k$ and ${\rm in}_{\prec}(f_{1})\succ{\rm in}_{\prec}(f_{i})$ for $i>k$. We set $h_{i} := f_{1}-f_{i}$ for $i=2,\dots,k$ and $h_{i} := f_{i}$ for $i=k+1,\dots,r$. Note that ${\rm in}_{\prec}(f_{1})\succ{\rm in}_{\prec}(h_{i})$ for $i\geq2$, $h_{2},\dots,h_{r}$ are in $\mathcal{L}^{1}\setminus\mathcal{L}^{2}$, and $h_{2},\dots,h_{r}$ are linearly independent over $K$. By applying the induction hypothesis to the set $H=\lbrace h_{2}\dots,h_{r}\rbrace$, we get a set $G'=\lbrace g_{2}\dots,g_{r}\rbrace\subseteq \mathcal{L}^{1}\setminus\mathcal{L}^{2}$ such that $KH=KG'$,  ${\rm in}_{\prec}(g_{2}),\dots,{\rm in}_{\prec}(g_{r})$ are distinct, and ${\rm in}_{\prec}(h_{i})\succeq{\rm in}_{\prec}(g_{i})$ for all $i\geq2$. Setting $g_{1} := f_{1}$ and $G := G'\cup\lbrace g_{1}\rbrace$, we get $KF=KG$, and the monomial ${\rm in}_{\prec}(g_{1})$ is distinct from ${\rm in}_{\prec}(g_{2}),\dots,{\rm in}_{\prec}(g_{r})$ because ${\rm in}_{\prec}(f_{1})\succ{\rm in}_{\prec}(h_{i})\succeq{\rm in}_{\prec}(g_{i})$ for $i\geq2$.
\end{itemize}

(b) By hypothesis, permuting if necessary, we have that ${\rm in}_{\prec}(f_{1})\succ\dots\succ{\rm in}_{\prec}(f_{r})$. Since $f_{1},\dots f_{r}$ $\in$ $\mathcal{L}^{1}\setminus\mathcal{L}^{2}$, it is sufficient to show that $f_{1},\dots f_{r}$ are linearly independent. Assume that $\sum_{i=1}^{r}\lambda_{i}f_{i}=0$, $\lambda_{i}\in K$ for all $i$. We proceed by contradiction assuming $\lambda_{1}=\dots=\lambda_{k-1}=0$ and $\lambda_{k}\neq0$ for some $k$. Setting $f=\sum_{i=k}^{r}\lambda_{i}f_{f}$, we get ${\rm in}_{\prec}(f)={\rm in}_{\prec}(f_{k})$ and $f\neq0$, a contradiction.

(c) This follows from (a) and (b).
\end{proof}

The following result describes how to compute the support of a code $D\subseteq K^m.$ For completeness, we also add a short proof.

\begin{lemma}{\cite[Lemma~2.1]{rth-footprint}}\label{seminar} Let $D\subseteq K^m$
be a subcode of dimension $r\geq 1$. If  
$\beta_1,\ldots,\beta_r$ is a $K$-basis for $D$ with
$\beta_i=(\beta_{i,1},\ldots,\beta_{i,m})$ for $i=1,\ldots,r$, then 
$\chi(D)=\bigcup_{i=1}^r\chi(\beta_i)$ and the number of elements of
$\chi(D)$ is the number of non-zero columns of the matrix:
$$   
\left[\begin{matrix}
\beta_{1,1}&\cdots&\beta_{1,i}&\cdots&\beta_{1,m}\\
\beta_{2,1}&\cdots&\beta_{2,i}&\cdots&\beta_{2,m}\\
\vdots&\cdots&\vdots&\cdots&\vdots\\
\beta_{r,1}&\cdots&\beta_{r,i}&\cdots&\beta_{r,m}
\end{matrix}\right].
$$
\end{lemma}
\begin{proof}
Let $j\in[m]$ be an element in $\chi(D)$. There exists $a=(a_{1},\dots,a_{m})\in D$ and $k_{1},\dots,k_{r}\in K$ such that $a=k_{1}\beta_1+\cdots+k_{r}\beta_r$ and $a_{j}=k_{1}\beta_{1,j}+\cdots+k_{r}\beta_{r,j}\neq0$. Thus, for some $i\in[r],$ $\beta_{i,j}\neq0.$ This means that $j\in\chi(\beta_{i})\subseteq \bigcup_{i=1}^r\chi(\beta_i)$. The other contention $\bigcup_{i=1}^r\chi(\beta_i) \subseteq \chi(\beta_{i})$
follows from the fact that $\beta_{i}\in D$.
\end{proof}

We come to one of the main results of this section. The following theorem gives an algebraic description of the relative generalized Hamming weights of standard evaluation codes. This description is useful for computations with {\it Macaulay2}~\cite{mac2} and the coding theory package~\cite{cod_package}.

\begin{theorem}\label{relative-Hamming-weights}
Let $X$ be a subset of $\mathbb{A}^{s}$ and $I$ the vanishing ideal of $X$. Given $\mathcal{L}^{2}\subsetneq\mathcal{L}^{1}$ two linear subspaces of $K\Delta_{\prec}(I)$, let $\mathcal{L}^{2}_{X}\subsetneq\mathcal{L}^{1}_{X}$ be the standard evaluation codes on $X$ relative to $\prec$. Then, 
for $r=1,\dots,{\rm dim}(C_{1})-{\rm dim}(C_{2})$, the relative generalized Hamming weights are given by
$$
M_{r}(\mathcal{L}_{X}^{1},\mathcal{L}_{X}^{2})={\rm deg}(S/I)-\max\lbrace{\rm deg}\left( S/(I,F)\right)|F\in\mathcal{R}_{\prec,r}\rbrace.
$$
\end{theorem}
\begin{proof}
Recall that by Equation (\ref{21.11.13}),
\[M_{r}(\mathcal{L}_{X}^{1},\mathcal{L}_{X}^{2})=\min\lbrace |\mathcal{X}(D)|: D\subseteq \mathcal{L}_{X}^{1}; D\cap \mathcal{L}_{X}^{2}=\lbrace0\rbrace,{\rm dim}(D)=r\rbrace.\]
Enumerate the points of $X=\{P_{1},\dots,P_{m} \}.$ Let $D$ be a subcode of $\mathcal{L}^{1}_{X}$ of dimension $r$ such that $D\cap\mathcal{L}^{2}_{X}=\lbrace0\rbrace$. The evaluation map induces an isomorphism of $K$-vector spaces between $\mathcal{L}^{1}$ and $\mathcal{L}^{1}_{X}$. Hence, by Lemma \ref{seminar}, there is a set $F=\lbrace f_{1},\dots,f_{r}\rbrace$ of linearly independent elements of $\mathcal{R}_{r}$, such that $D=\bigoplus_{i=1}^{r}K\alpha_{i}$, where $\alpha_{i} := (f_{i}(P_{1}),\dots,f_{i}(P_{m}))$, and the support $\mathcal{X}(D)$ is equal to $\bigcup_{i=1}^{r}\mathcal{X}(\alpha_{i}).$ Notice that $j\in\mathcal{X}(D)=\bigcup_{i=1}^{r}\mathcal{X}(\alpha_{i})$ if and only if  $j\in\mathcal{X}_{\alpha_{i}}$ for some $i$, if and only if $P_{j}\in X\setminus V_{X}(F)$. Therefore, we get 
$$
|\mathcal{X}(D)|=|X\setminus V_{X}(F)|.
$$
Conversely let $F=\lbrace f_{1}\dots,f_{r}\rbrace$ be a set of $\mathcal{R}_{r}$, then there is a subcode $D$ of $\mathcal{L}^{1}_{X}$ of dimension $r$ such that $D\cap\mathcal{L}_{X}^{2}=\lbrace0\rbrace$. Setting 
$$
\alpha_{i} := (f_{i}(P_{1}),\dots,f_{i}(P_{m}))\quad \text{for}\quad i=1,\dots,r,
$$
then $D=K\alpha_{i}+\dots+K\alpha_{r}$. Thus, we obtain that $|\mathcal{X}(D)|=|X\setminus V_{X}(F)|.$ For an element $F\in \mathcal{R}_{r}$, Lemma \ref{quotient-space-lemma} (c) implies that there is $G\in\mathcal{R}_{\prec,r}$ such that $KF=KG$. Thus, $V_{X}(F)=V_{X}(G)$, and we conclude that 
$$
\lbrace V_{X}(F)|F\in\mathcal{R}_{r}\rbrace\subseteq\lbrace V_{X}(F)|F\in\mathcal{R}_{\prec,r}\rbrace.
$$
On the other hand, by Lemma \ref{quotient-space-lemma} (c), we know that $\mathcal{R}_{\prec,r}\subseteq\mathcal{R}_{r}$, then the affine varieties defined by the elements of $\mathcal{R}_{r}$ and $\mathcal{R}_{\prec,r}$ are the same
$$
\lbrace V_{X}(F)|F\in\mathcal{R}_{r}\rbrace=\lbrace V_{X}(F)|F\in\mathcal{R}_{\prec,r}\rbrace.
$$
Hence, we obtain
\begin{align*}
M_{r}(\mathcal{L}^{1},\mathcal{L}^{2})&=\min\lbrace |\mathcal{X}(D)|:D\subseteq\mathcal{L}_{X}^{1};D\cap\mathcal{L}_{X}^{2}=\lbrace0\rbrace,{\rm dim}(D)=r\rbrace\\
&=\min\lbrace |X\setminus V_{X}(F)|: F\in\mathcal{R}_{r}\rbrace\\
&=|X|-\max\lbrace|V_{X}(F)|:F\in\mathcal{R}_{r}\rbrace\\
&={\rm deg}\left(S/I \right) -\max\lbrace|V_{X}(F)|:F\in\mathcal{R}_{\prec,r}\rbrace\\
&={\rm deg}\left(S/I \right) -\max\lbrace{\rm deg}\left(S/(I,F) \right):F\in\mathcal{R}_{\prec,r}\rbrace.
\end{align*}
Which proves the theorem.
\end{proof}

Note that if $\mathcal{L}^{2}=\lbrace0\rbrace$ and $\mathcal{L}_{\prec,r}$ denotes the set of all subsets $F=\lbrace f_{1},\dots,f_{r}\rbrace$ of size $r$ of $\mathcal{L}^{1}\setminus\lbrace0\rbrace$ such that ${\rm in}_{\prec}(f_{1}),\dots,{\rm in}_{\prec}(f_{r})$ are distinct monomials and $f_{i}$ is monic for $i=1,\dots,r,$ then $\mathcal{L}_{\prec,r}=\mathcal{R}_{\prec,r}$. Thus, as a consequence of Remark \ref{case-of-gener-Ha-we} and Theorem \ref{relative-Hamming-weights} we recover the following description for the generalized Hamming weights of evaluation codes.

\begin{corollary}{\cite[Theorem 3.4]{Jaramillo-Vaz-pinto-Villarreal}}
Let $X$ be a subset of $\mathbb{A}^s$, let $I$ be the vanishing
ideal of $X$, let $\mathcal{L}$ be a linear
subspace of $K\Delta_\prec(I)$, and let $\mathcal{L}_X$ be the standard evaluation code
on $X$ relative to $\prec$. Then 
$$
\delta_r(\mathcal{L}_X)=\deg(S/I)-\max\{\deg(S/(I,F))\:\vert\:\,
F\in\mathcal{L}_{\prec,r}\}\ \mbox{ for }\ 1\leq r\leq\dim_K(\mathcal{L}_X).
$$
\end{corollary}

Observe that Theorem~\ref{relative-Hamming-weights} gives a description of relative generalized Hamming weights of standard evaluation codes. The following result, proved in \cite[Proposition 3.5.]{Jaramillo-Vaz-pinto-Villarreal} and then in \cite[Corollary 3.2]{Lopez-Soprunov-Villarreal}, implies that actually, Theorem~\ref{relative-Hamming-weights} can be used for any evaluation code, rather than only standard evaluation codes. For completeness, we also include a proof here.
\begin{proposition}\label{Standar-vs-evaluation-codes}
Let $\mathcal{L}_{X}$ be an evaluation code on $X$ and let $\prec$ be a monomial order. Then there exists a unique linear subspace $\widetilde{\mathcal{L}}$ of $K\Delta_{\prec}(I)$ such that $\widetilde{\mathcal{L}}_{X}=\mathcal{L}_{X}$.
\end{proposition}
\begin{proof}
Let $X$ be a set of points
$\{P_1,\ldots,P_m\}$ of the affine space 
$\mathbb{A}^s$ and let $\mathcal{G}$ be a Gr\"obner basis of
$I=I(X)$. Pick a $K$-basis $\{f_1,\ldots,f_k\}$ of $\mathcal{L}$.
For each $i$, let $r_i$ be the remainder on division of $f_i$ by 
$\mathcal{G}$, that is, by the division algorithm \cite[Theorem~3,
p.~63]{CLO}. For each $i$, we can write $f_i=h_i+r_i$, where $h_i\in
I$, and $r_i=0$ or $r_i$ is a standard polynomial of $S/I$. We set
$$
\widetilde{\mathcal{L}}:=Kr_1+\cdots+Kr_k.
$$
\quad The evaluation code $\widetilde{\mathcal{L}}_X$ is a standard
evaluation code on $X$ relative to $\prec$ since
$\widetilde{\mathcal{L}}$ is a linear subspace of $K\Delta_\prec(I)$.
To show the inclusion
${\mathcal{L}}_X\subseteq\widetilde{\mathcal{L}}_X$ take a point $P$ in
${\mathcal{L}}_X$. Then, $P$ is equal to $(f(P_1),\ldots,f(P_m))$ for some
$f\in\mathcal{L}$. Using the equations $f_i=h_i+r_i$, $i=1,\ldots,k$,
we can write $f=h+r$, where $h\in I$ and
$r\in\widetilde{\mathcal{L}}$. Hence, $P$ is equal to
$(r(P_1),\ldots,r(P_m))$, that is, $P\in \widetilde{\mathcal{L}}_X$. 
By other hand, for 
$\widetilde{\mathcal{L}}_X\subseteq{\mathcal{L}}_X$ take a point $Q$ in 
$\widetilde{\mathcal{L}}_X$. Then, $Q$ is equal to $(g(P_1),\ldots,g(P_m))$ for some
$g\in\widetilde{\mathcal{L}}$. Using the equations $f_i=h_i+r_i$, $i=1,\ldots,k$,
we can write $g=g_1+g_2$, where $g_1\in I$ and
$g_2\in{\mathcal{L}}$. Hence, $Q$ is equal to
$(g_2(P_1),\ldots,g_2(P_m))$, that is, $Q\in{\mathcal{L}}_X$.
\end{proof}
\rmv{
\begin{corollary}\label{relative-linear-codes}
Let $C_{2}\subsetneq C_{1}$ be linear codes. Then there exists a set $X$,  and linear subspaces $\mathcal{L}^{2}\subset\mathcal{L}^{1}$ of $S$ such that  
$$
M_{r}(C_{1},C_{2})=M_{r}(\mathcal{L}^{1}_{X},\mathcal{L}^{2}_{X})={\rm deg}(S/I)-\max\lbrace{\rm deg}\left( S/(I,F)\right)|F\in\mathcal{R}_{\prec,r}\rbrace.
$$
\end{corollary}

\begin{proof}
This follows from Proposition \ref{Standar-vs-evaluation-codes}, Proposition \ref{linear-code-then-evalu-code}, and Theorem \ref{relative-Hamming-weights}
\end{proof}
}

The purpose of the following lines is to give a bound for the relative generalized Hamming weights of evaluation codes. Let $\mathcal{M}_{\prec,r}$ be the family of all subsets $M$ of ${\rm in}_{\prec}(\mathcal{L}^{1}\setminus\mathcal{L}^{2})=\lbrace{\rm in}_{\prec}(f): f\in \mathcal{L}^{1}\setminus\mathcal{L}^{2}\rbrace$ with $r$ distinct elements.

\begin{definition}
The $r$-th {\it relative footprint} of the standard evaluation codes $\mathcal{L}_{X}^{2}\subsetneq\mathcal{L}_{X}^{1}$ is denoted and defined by
$$
{\rm RFP}_{r}(\mathcal{L}_{X}^{1},\mathcal{L}_{X}^{2}):={\rm deg}\left(S/I \right) -\max\lbrace{\rm deg}\left(S/({\rm in}_{\prec}(I),M) \right):M\in\mathcal{M}_{\prec,r}\rbrace.
$$
\end{definition} 

\begin{theorem}\label{foorprint-relative-gen-weight}
Let $X$ be a subset of $\mathbb{A}^{s}$ and $I$ the vanishing ideal of $X$. Given $\mathcal{L}^{2}\subsetneq\mathcal{L}^{1}$ two linear subspaces of $K\Delta_{\prec}(I)$, let $\mathcal{L}^{2}_{X}\subsetneq\mathcal{L}^{1}_{X}$ be the standard evaluation codes on $X$ relative to $\prec$. Then,
\begin{equation}\label{ineq-foot}
{\rm RFP}_{r}(\mathcal{L}_{X}^{1},\mathcal{L}_{X}^{2})\leq M_{r}(\mathcal{L}_{X}^{1},\mathcal{L}_{X}^{2})\quad\text{for}\quad 1\leq r\leq {\rm dim}_{K}(\mathcal{L}_{X}^{1})-{\rm dim}_{K}(\mathcal{L}_{X}^{2}).
\end{equation}
\end{theorem}

\begin{proof}
By Theorem \ref{relative-Hamming-weights}, we have that
$$
M_{r}(\mathcal{L}_{X}^{1},\mathcal{L}_{X}^{2})={\rm deg}(S/I)-\max\lbrace{\rm deg}\left( S/(I,F)\right)|F\in\mathcal{R}_{\prec,r}\rbrace.
$$
Then we need to show only that 
$$
{\rm deg}\left(S/I \right) -\max\lbrace{\rm deg}\left(S/({\rm in}_{\prec}(I),M) \right):M\in\mathcal{M}_{\prec,r}\rbrace\leq {\rm deg}(S/I)-\max\lbrace{\rm deg}\left( S/(I,F)\right)|F\in\mathcal{R}_{\prec,r}\rbrace.
$$
Previous inequality is equivalent to 
$$
\max\lbrace{\rm deg}\left( S/(I,F)\right)|F\in\mathcal{R}_{\prec,r}\rbrace\leq \max\lbrace{\rm deg}\left(S/({\rm in}_{\prec}(I),M) \right):M\in\mathcal{M}_{\prec,r}\rbrace.
$$
This inequality follows from the fact that for each $F\in\mathcal{R}_{\prec,r}$, ${\rm in}_{\prec}(F)\in \mathcal{M}_{\prec,r}.$ Thus, from Theorem  \ref{degree-initial-theorem}, we obtain
$$
{\rm deg}(S/\left( I,F\right) )\leq{\rm deg}(S/\left( {\rm in}_{\prec}(I),{\rm in}_{\prec}(F)\right) ).
$$
Which proves the result.
\end{proof}

The  bound in Theorem \ref{foorprint-relative-gen-weight} is sharp. Indeed, there exist families of evaluation code where the equality holds. For instance, for the affine Cartesian codes, which we describe now. Let $d_{1}\leq\dots\leq d_{s}$ be positive integers and $A_{1},\dots,A_{s}$ a sequence of subsets of $K$ with cardinalities $d_{1},\dots,d_{s},$ respectively. Denote by $\mathcal{A}$ the Cartesian product $\mathcal{A}:=A_{1}\times\dots\times A_{s}$. Note that $m=|\mathcal{A}|=d_{1}\times\dots\times d_{s}$. Enumerate the elements of $\mathcal{A}=\{P_{1},\dots,P_{m}\}.$ For a positive integer $d\leq \sum_{i=1}^{s}(d_{i}-1)$, we define the subspace 
$$
S_{\leq d}(\mathcal{A}):=\lbrace f\in S :\deg_{t_{i}}(f)<d_{i}, \deg(f)\leq d\rbrace.
$$
The map 
$$
{\rm ev}_\mathcal{A}:S_{\leq d}(\mathcal{A})\rightarrow K^{m}, \quad
f\mapsto(f(P_{1}),\dots,f(P_{m})),
$$
is a linear map. The image $C_{\mathcal{A}}(d) := {\rm ev}(S_{\leq d}(\mathcal{A}))$ is called the {\it affine Cartesian code} of degree $d$.
The relative generalized Hamming weights of an affine Cartesian code with respect to a smaller affine Cartesian code have been computed by M. Datta {\rm \cite{datta}} in the following way. Take the set
$$
\mathcal{F} := \lbrace 0,\dots,d_{1}-1\rbrace\times\dots\times\lbrace 0,\dots,d_{s}-1\rbrace.
$$
For an element ${\bf a}=(a_{1},\dots,a_{s})\in \mathcal{F}$, we define ${\rm deg}({\bf a}) := a_{1}+\cdots+a_{s}$. We introduce some subsets of $\mathcal{F}$ consisting of elements satisfying certain degree constraints. For any integer $d$, define 
$$
\mathcal{F}_{\leq d}:=\lbrace {\bf a}\in \mathcal{F}: \deg({\bf a})\leq d\rbrace. 
$$
On a similar way, for integers $d_{1}$, $d_{2}$ satisfying $d_{2}<d_{1}$, we define 
$$
\mathcal{F}_{d_{2}}^{d_{1}}:=\lbrace {\bf a}\in \mathcal{F}: d_{2}<\deg({\bf a})\leq d_{1}\rbrace. 
$$
\begin{theorem}{\rm \cite[Theorem 4.3, Corollary 3.10]{datta}}
Fix integers $d_{1},d_{2}$ with $-1\leq d_{2}<d_{1}\leq\sum_{i=1}^{m}(d_{i}-1)$. Let $C_{\mathcal{A}}(d_{1})$ and $C_{\mathcal{A}}(d_{2})$ denote the corresponding affine Cartesian codes. For any integer $1\leq r\leq\dim(C_{\mathcal{A}}(d_{1}))-\dim(C_{\mathcal{A}}(d_{2}))$, the $r$-th RGHW and relative footprint bound of $C_{\mathcal{A}}(d_{1})$ with respect to $C_{\mathcal{A}}(d_{2})$ is given by 
$$
M_{r}(C_{\mathcal{A}}(d_{1}),C_{\mathcal{A}}(d_{2}))=d_{1}\cdots d_{s}-\sum_{i=1}^{s}a_{r,i}\prod_{j=i+1}^{s}d_{j}-t+r={\rm RFP}_{r}(C_{\mathcal{A}}(d_{1}),C_{\mathcal{A}}(d_{2})),
$$
where $(a_{r,1},\dots,a_{r,s})$ is the r-th element of $\mathcal{F}_{d_{2}}^{d_{1}}$ and t-th element of $\mathcal{F}_{\leq d_{1}}$ in descending lexicographic order.
\end{theorem}

The square free evaluation codes is another family of evaluation codes where the equality in Theorem \ref{foorprint-relative-gen-weight} holds. The \textit{affine torus} of the affine space $\mathbb{A}^s:=K^s$ is 
given by $T:=(K^*)^s=\lbrace P_{1},\dots,P_{m}\rbrace$, where $K^*$ is the multiplicative group of $K$. Let $V_{\leq d}$ be the set of all squarefree monomials of $S$ of degree at most $d$. The image of the evaluation map,
$$
{\rm ev}:KV_{\leq d}\rightarrow K^{m},\quad f\mapsto (f(P_{1}),\dots,f(P_{m})),
$$
denoted by $\mathcal{C}_{\leq d}$, is called the \textit{squarefree evaluation code} of degree $d.$

Let $\mathcal{C}_{\leq d_{2}} \subsetneq \mathcal{C}_{\leq d_{1}}$ be two squarefree evaluation codes with $1\leq d_{2}<d_{1}\leq s$. By {\rm \cite[Proposition 5.1]{Jaramillo-Vaz-pinto-Villarreal}}, we have that 
$$
\max_{f\in KV_{\leq d_{2}}}|V_{T}(f)|<\max_{f\in KV_{\leq d_{1}}}|V_{T}(f)|.
$$
Then, the code $\mathcal{C}_{\leq d_{2}}$ has no codewords of minimum weight. Thus, we conclude that 
$$
M_{1}(\mathcal{C}_{\leq d_{1}},\mathcal{C}_{\leq d_{2}})=\delta(\mathcal{C}_{\leq d_{1}}),
$$
where $\delta(\mathcal{C}_{\leq d_{1}})$ denotes minimum distance of $\mathcal{C}_{\leq d_{1}}.$ As {\rm \cite[Theorem 5.5]{Jaramillo-Vaz-pinto-Villarreal}} implies that the minimum distance of $\mathcal{C}_{\leq d_{1}}$ is given by the footprint bound, then  
$$
{\rm RFP}_{1}(\mathcal{C}_{\leq d_{1}},\mathcal{C}_{\leq d_{2}})=M_{1}(\mathcal{C}_{\leq d_{1}},\mathcal{C}_{\leq d_{2}}).
$$
In Example \ref{sharp-bound-inequ}, we present two linear standard evaluation codes where the inequality of Theorem \ref{foorprint-relative-gen-weight} is strict.

It is important to note that the degree of $I(X)$ has been used before to describe some properties of evaluation codes. Here is a brief summary to put in context where this work is positioned.
\setlist[enumerate]{leftmargin=0.5 cm}
\begin{enumerate}
\item Minimum distance function. In~\cite{Bernal-Pitones-Villarreal2}, the degree was used to introduced the minimum distance function. This is a function associated to an ideal and depends on an integer $d$. When the ideal is the vanishing ideal $I(X),$ then the minimum distance function at $d$ coincides with the minimum distance of the evaluation code $C_X(d).$ These functions are studied further in \cite{Bernal-Pitones-Villarreal2} and \cite{NB-YP-RV}.
\item Generalized minimum distance function. In~\cite{rth-footprint}, the degree was used to introduced the generalized minimum distance function. In a similar way to~\cite{Bernal-Pitones-Villarreal2}, this functions agrees to the generalized Hamming weight of the evaluation code $C_X(d)$ when the function is associated to the vanishing ideal $I(X)$ and evaluated on the integer $d.$
\item Relative generalized minimum distance function~\cite{GUSR}. This function is associated to an ideal and depends on two integers $d_1$ and $d.$ When $d_1=0,$ this function agrees with the generalized minimum distance function. In addition, when the ideal is the vanishing ideal, this function matches with the relative generalized Hamming weights of the evaluation codes $C_X(d_1) \subseteq C_X(d).$ This function is also defined for points $X$ in a projective space.
\item Generalized Hamming weights of evaluation codes. In~\cite{Jaramillo-Vaz-pinto-Villarreal}, the degree is used to algebraically describe the generalized Hamming weights of evaluation codes.
\item Relative generalized Hamming weights of evaluation codes. In this work, the degree is used to algebraically describe the relative generalized Hamming weights of the evaluation codes $C_X(d_1) \subseteq C_X(d).$
\end{enumerate}
It is important to remark the relations between previous items. Item (3) generalizes (2), which is generalizing (1). Item (5) generalizes (4). Items (2) and (4) are related, they both agree when the evaluation code is a Reed-Muller-type code, but neither generalizes the other. In a similar way, (5) and (3) agrees sometimes, but neither generalizes the other.

\section{Toric codes over hypersimplices}\label{sec-weight-distribution}
In this section, we introduces toric codes over hypersimplices and compute the next-to minimal weights for certain cases. 

Take an integer $1\leq d\leq s.$ Let $\mathcal{P}$ be the convex hull in $\mathbb{R}^s$
of all integral points $\mathbf{e}_{i_1}+\cdots+\mathbf{e}_{i_d}$
such that $1\leq i_1<\cdots< i_d\leq 
s$, where $\mathbf{e}_i$ is the $i$-th unit vector in $\mathbb{R}^s$. The lattice polytope
$\mathcal{P}$ is called the $d$-th \textit{hypersimplex} of
$\mathbb{R}^s$ \cite[p.~84]{Stur1}.  
 The \textit{toric code} of
$\mathcal{P}$ of degree $d$,  
denoted $\mathcal{C}_\mathcal{P}(d)$ or simply
$\mathcal{C}_d$, is the image 
of the evaluation map 
\begin{equation}\label{sep13-21}
{\rm ev}_d\colon KV_d\rightarrow K^{m},\quad
f\mapsto\left(f(P_1),\ldots,f(P_m)\right),
\end{equation}
where $KV_d$ is the $K$-linear subspace of $S_d$ spanned by
the set $V_d$ of all $t^a$ such that $a\in\mathcal{P}\cap\mathbb{Z}^s$, and
$\{P_1,\ldots,P_m\}$ is the set of all points 
of the affine torus $T=(K^*)^s$. Observe that a monomial $t^a$ of $S$ is in
$KV_d$ if and only if $t^a$ is squarefree and has degree $d$. The set
$V_d$ is precisely the set of squarefree monomials of $S$ of degree $d$.
The following result gives the minimum distance of $\mathcal{C}_{\mathcal{P}}(d)$.

\begin{theorem}\cite[Theorem 4.5]{Jaramillo-Vaz-pinto-Villarreal}\label{minimum-distance-hypersimplex} 
Let $\mathcal{C}_\mathcal{P}(d)$ be the toric code of $\mathcal{P}$
of degree $d$ and let $\delta(\mathcal{C}_\mathcal{P}(d))$
be its minimum distance. Then 
$$
\delta(\mathcal{C}_\mathcal{P}(d))=\begin{cases}
(q-2)^d(q-1)^{s-d}&\mbox{ if }d\leq s/2,\, q\geq 3,\\
(q-2)^{s-d}(q-1)^{d}&\mbox{ if }s/2 < d < s, \,q\geq 3,\\
(q-1)^{s} &\mbox{ if } d=s,
\\
1 &\mbox{ if } q=2.
\end{cases}
$$
\end{theorem}
 

To prove the next proposition we use the results of
Section~\ref{Prelim-section}.   
\begin{proposition}\cite[Proposition 4.3]{Jaramillo-Vaz-pinto-Villarreal}\label{squarefree-affine} Let $f$ be a  squarefree homogeneous of degree $d$  and let $T$ be the affine torus $(K^{*})^s$ of $\mathbb{A}^s$. If $q\geq
3$ and $1\leq d< s$, then 
$$
|V_T(f)|\leq(q-1)^{s}-(q-2)^{d}(q-1)^{s-d}.
$$
\end{proposition}

\begin{proof}
Let $\prec$ be a monomial order on $S$ and let $I=I(T)$
be the vanishing ideal of $T$. 
By Theorem~\ref{degree-initial-theorem}, one has
\begin{equation}\label{apr17-1}
|V_T(f)|=\deg(S/(I,f)).
\end{equation}
\quad The initial ideal $L:={\rm
in}_\prec(I)$ of $I$ is
generated by the set $\{t_i^{q-1}\}_{i=1}^{s}$. Let
$t^a={\rm in}_\prec(f)=t_1^{a_1}\cdots 
t_s^{a_s}$ be the initial monomial of $f$.  Since $f$ is monomially squarefree, so is
$t^a$. As $q\geq 3$, $t^a$ cannot be in $L$. Therefore, by
Theorem~\ref{degree-initial-theorem} and 
Lemma~\ref{apr5-19}, we get  
\begin{equation}\label{apr17-2}
\deg(S/(I,f))\leq(q-1)^{s}-(q-2)^{d}(q-1)^{s-d},
\end{equation}
where $d=\deg(f)$. Thus, the inequality follows at once from
Eq.~(\ref{apr17-1}).
\end{proof}

\begin{corollary}\label{bounds-homo-square}
Let $f_{d}$ be a reducible squarefree homogeneous polynomial  of degree $d$,
$$
f_{d}=h_{r}t_{i_{r+1}}\cdots t_{i_{r+d}}
$$
with $r<d$ and ${\rm deg}(h_{r})=r$. If $q\geq
3$ and $1<d< s$, then 
$$
|V_T(f_{d})|\leq(q-1)^{s}-(q-2)^{r}(q-1)^{s-r}.
$$
\end{corollary}

\begin{proof}
The polynomial $h_{r}$ is a squarefree homogeneous polynomial in the variable $\lbrace t_{1},\dots, t_{s}\rbrace\setminus\lbrace t_{i_{r+1}},\dots,t_{i_{r+d}}\rbrace$, because the polynomial $f_{d}$ is a squarefree homogeneous polynomial. Notice that $V_{T}(f_{d})$ is equal to $V_{T}(h_{r})$, using Proposition \ref{squarefree-affine} we conclude that 
$$
|V_T(f_{d})|=|V_T(g_{r})|\leq(q-1)^{s}-(q-2)^{r}(q-1)^{s-r}.
$$
\end{proof}

\begin{proposition}\label{weights-degree-one}Let $r$ be an integer such that $2\leq r\leq s$. If $\alpha_{j_{1}},\dots,\alpha_{j_{r}}\in K^*$, and $f_{r}=\alpha_{j_{1}}t_{i_{1}}+\cdots+\alpha_{j_{r}}t_{i_{r}}$, then 
$$
|V_{T}(f_{r})|=(q-1)^{s-1}-(q-1)^{s-2}+\cdots+(-1)^{r}(q-1)^{s-(r-1)}.
$$
\end{proposition}
\begin{proof}
We use induction on $s$. For $s=2$, the number of solution over $T$ is $(q-1)$. This is because $\alpha_{j_{1}}t_{1}+\alpha_{j_{2}}t_{2}=0$ is equivalent to $t_{1}=-\frac{\alpha_{j_{2}}}{\alpha_{j_{1}}}t_{2}$. Assume $K^{*}=\lbrace\alpha_{1},\alpha_{2},\dots,\alpha_{q-1}\rbrace$ and  $T':=(K^{*})^{s-1}$. Define the set
$$
B:=\lbrace (\alpha_{1},\alpha_{2},\dots,\alpha_{i_{1}-1},0,\alpha_{i_{i}+1},\dots,\alpha_{q-1}): a'_{2}\alpha_{i_{2}}+\cdots+a'_{r}\alpha_{i_{r}}=0\rbrace.
$$
The equation 
\begin{equation}\label{sum-m-variables}
\alpha_{j_{1}}t_{i_{1}}+\cdots+\alpha_{j_{r}}t_{i_{r}}=0
\end{equation} 
is equivalent to the equation 
$t_{i_{1}}=\alpha'_{j_{2}}t_{i_{2}}+\cdots+\alpha'_{j_{r}}t_{i_{r}}.$
Thus, 
the number of solution of Eq.~(\ref{sum-m-variables}) is 
$$
(q-1)^{s-1}-|B|.
$$
Given the bijection
\begin{align*}
B&\rightarrow V_{T'}(\alpha'_{j_{2}}t_{i_{2}}+\cdots+\alpha'_{j_{r}}t_{i_{r}})\\
(\alpha_{1},\alpha_{2},\dots,\alpha_{i_{1}-1},0,\alpha_{i_{i}+1}&,\dots,\alpha_{q-1})\mapsto(\alpha_{1},\alpha_{2},\dots,\alpha_{i_{1}-1},\alpha_{i_{i}+1},\dots,\alpha_{q-1}),
\end{align*}
we obtain that
$
|V_{T}(f)|=(q-1)^{s-1}-|V_{T'}(\alpha'_{j_{2}}t_{i_{2}}+\cdots+\alpha'_{j_{r}}t_{i_{r}})|.
$
Therefore, by induction,
\begin{align*}
|V_{T}(f)|=&(q-1)^{s-1}-\left( (q-1)^{(s-1)-1}-(q-1)^{(s-1)-2}+\cdots+(-1)^{r-1}(q-1)^{(s-1)-((r-1)-1)}\right)\\
=&(q-1)^{s-1}-(q-1)^{s-2}+(q-1)^{s-3}+\cdots+(-1)^{r}(q-1)^{s-(r-1)}.
\end{align*}
Which proves the proposition.
\end{proof}

We come to the main result of this section.
\begin{theorem}\label{theo-weight-dist}
Let $\mathcal{C}_{\mathcal{P}}(1)$ be the toric code of $\mathcal{P}$ of degree $1$. For $1\leq t\leq s/2,$ the $t$-th Hamming weight of $\mathcal{C}_{\mathcal{P}}(1)$ is given by
$$
(q-1)^{s}-(q-1)^{s-1}+(q-1)^{s-2}+\cdots+(-1)^{2t+1}(q-1)^{s-(2t-1)}.
$$
\end{theorem}
\begin{proof}
From Proposition \ref{weights-degree-one},
$$
|V_{T}(f_{m})|=(q-1)^{s-1}-(q-1)^{s-2}+\cdots+(-1)^{m}(q-1)^{s-(m-1)},
$$
where $f_{m}=a_{1}t_{i_{1}}+\cdots+a_{m}t_{i_{m}}\in KV_{1}$ is a homogeneous squarefree of degree one. 
If $1\leq t\leq s/2$, notice that 
\begin{align*}
|V_{T}(f_{2(t+1)})|&=|V_{T}(f_{2t})|-(q-1)^{s-(2t)}+(q-1)^{s-(2t+1)}<|V_{T}(f_{2t})|,\\
|V_{T}(f_{2t+1})|&=|V_{T}(f_{2t})|-(q-1)^{s-(2t)}<|V_{T}(f_{2t})|, \text{ and }\\
|V_{T}(f_{2t+1})|&=|V_{T}(f_{2t})|-(q-1)^{s-(2t)}\\
&<|V_{T}(f_{2t})|-(q-1)^{s-(2t)}+(q-1)^{s-(2t+1)}-(q-1)^{s-(2t+2)}=|V_{T}(f_{2t+3})|.\\
\end{align*}
This implies that the weight distribution of $\mathcal{C}_{\mathcal{P}}(1)$ is given by the even values of $m.$ In other words, the $t$-th Hamming weight is 
\begin{align*}
|T|-&\left( (q-1)^{s-1}+(q-1)^{s-2}+\cdots+(-1)^{2t}(q-1)^{s-(2t-1)}\right)\\
=&(q-1)^{s}-(q-1)^{s-1}+(q-1)^{s-2}+\cdots+(-1)^{2t+1}(q-1)^{s-(2t-1)}.
\end{align*}
Thus, we obtain the result.
\end{proof}

\section{Examples}\label{example-sec}
In this section we include implementations in {\it Macaulay2} \cite{cod_package,mac2} and examples that illustrate some of the results of the previous sections.

\begin{example}\label{five-points}
Take $K=\mathbb{F}_{3}$ and $X=\lbrace (0,0), (1,0), (0,1), (1,1), (0,-1)\rbrace \subset K^{2}$. The vanishing ideal $I=I(X)$ of $X$ is generated by $t_{1}^{2}-t_{1},t_{2}^{3}-t_{2}, t_{1}t_{2}^{2}-t_{1}t_{2}$. We compute the first and second relative generalized Hamming weights for the Reed-Muller codes $C_{X}(d)$ of degree $d=1,2$. Using Theorem \ref{relative-Hamming-weights} and Procedure  \ref{procedure-five-points} we have
$$
M_{1}(C_{X}(2),C_{X}(1))=1,
$$
$$
M_{2}(C_{X}(2),C_{X}(1))=2.
$$
\end{example} 

\begin{example}\label{second-Hamming-weight-example}
Let $K$ be the field $\mathbb{F}_{3}$, $T$ the affine torus $(\mathbb{F}_{3}^{*})^{4}$, and $I=I(T)$ the vanishing ideal of $T$. Using Theorem \ref{minimum-distance-hypersimplex} and Procedure \ref{second-weight-procedure}, we obtain the following table:
\begin{eqnarray*}
\hspace{-11mm}&&\left.
\begin{array}{c|c|c|c|c}
d & 1 & 2 & 3 & 4\\
   \hline
m & 16 & 16 & 16 & 16\\
   \hline
\dim_K(\mathcal{C}_\mathcal{P}(d)) & 4 & 6&  4& 1
 \\ 
   \hline
 \delta(\mathcal{C}_\mathcal{P}(d))   \    & 8 & 4  & 8&16\\
    \hline
 \delta^{2}(\mathcal{C}_\mathcal{P}(d))   \    & 10 & 6  & 10&16\\
\end{array}
\right.
\end{eqnarray*}
\noindent
where $\delta^{2}(\mathcal{C}_\mathcal{P}(d))$ denotes de second Hamming weight of the toric code  $\mathcal{C}_\mathcal{P}(d)$.
\end{example}

\begin{example}\label{sharp-bound-inequ}
Let $K$ be the field $\mathbb{F}_{5}$ and $S=K[t_{1},t_{2}]$ the polynomial ring in two variables. Let $\mathcal{L}^{1}$ be the linear space generated by all monomials $t_{1}^{a_{1}}t_{2}^{a_{2}}$ such that $(a_{1},a_{2})$ is one of the solid points or plus symbol points of the configuration depicted in Figure $1$. Let $\mathcal{L}^{2}$ be the linear subspace of $\mathcal{L}^{1}$ generated only by all monomials $t_{1}^{a_{1}}t_{2}^{a_{2}}$ such that $(a_{1},a_{2})$ is one of plus symbol points. The vanishing ideal of the affine torus $T=(\mathbb{F}_{5}^{*})^{2}$ is generated by the Gr\"obner basis $G=\lbrace t_{1}^{4}-1,t_{2}^{4}-1\rbrace$. Then, the standard evaluation codes $\mathcal{L}_{T}^{1}$ and $\mathcal{L}_{T}^{2}$ are generated by 
$$
B^{1}=\lbrace 1,t_{3},t_{1}t_{2}^{2},t_{2}^{3},t_{1}t_{2},t_{1}^{2}\rbrace \quad \hbox{ and } \quad B^{2}=\lbrace t_{1}t_{2}^{2},t_{1}t_{2}\rbrace,
$$
respectively. Using Theorem \ref{relative-Hamming-weights}, Theorem \ref{foorprint-relative-gen-weight}, and Procedure \ref{sharp-inequ}, we obtain that the relative minimum distance $M_{1}(\mathcal{L}_{T}^{1},\mathcal{L}_{T}^{2})$ of $\mathcal{L}_{T}^{1}$ with respect to $\mathcal{L}_{T}^{2}$ is $8$ and the relative footprint bound $RFP_{1}(\mathcal{L}_{T}^{1},\mathcal{L}_{T}^{2})$ is $4$.
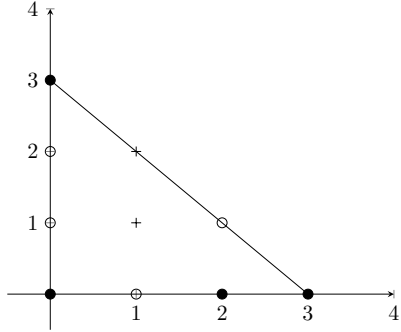
\begin{figure}[H]
\begin{tikzpicture}[scale = 0.75]
\begin{axis}[axis lines=center, enlargelimits=false, 
 xmin=-0.5,
  xmax=4,
  ymin=-0.5,
  ymax=4,
  title={}]
  \addplot [only marks, mark size=2.5] table {
  0 0
  0 3
  2  0
  3  0
};
 \addplot [only marks, mark=+, mark size=2.5] table {
  1  1
  1  2
};
\addplot [only marks, mark=o, mark size=2.5] table {
 2  1
 1  0
 0  1
 0  2
};
\addplot [domain=0:1, samples=2]{-1*x+3};
\addplot [domain=1:3, samples=2]{-1*x+3};
\end{axis}

\end{tikzpicture}
\caption{Lattice points defining $\mathcal{L}^{1}$ and
 $\mathcal{L}^{2}$ of Example~\ref{sharp-bound-inequ}.}\label{transforming-figure}
\end{figure}

\end{example}


\begin{procedure}\label{procedure-five-points}
Computing the relative generalized Hamming weights of a Reed-Muller-type code $C_X (d)$ using Theorem \ref{relative-Hamming-weights}. This procedure corresponds to Example \ref{five-points}. 

\begin{verbatim}
q=3, K=ZZ/q, S=K[t1,t2];
I1=ideal(t1,t2)
I2=ideal(t1-1,t2) 
I3=ideal(t1,t2-1)
I4=ideal(t1-1,t2-1)
I5=ideal(t1,t2+1) 
I=intersect(I1,I2,I3,I4,I5)
M=coker gens gb I
---
relgenfun = method(TypicalValue => ZZ);
 relgenfun (ZZ,ZZ,ZZ,Ideal) := (d,t,r,I) ->(
H=flatten entries basis(0,d,M);     
J=apply(toList (set(0..char ring I-1))^**(#H)-(set{0})^**(#H),toList);
Y=apply(J,x->matrix {H}*vector deepSplice x);
F=apply(Y,z->flatten entries z);
W=toList(apply(F,x-> if (degree x#0)#0 > t then x#0 else 0)-set{0});
V=apply(W,m->(leadCoefficient(m))^(-1)*m);
C=subsets(V,r);
B=apply(C,x->ideal(x)); 
A=apply(B,x -> if #(set flatten entries leadTerm gens x)==r then 
degree(I+x) else 0);
degree M-max A 
)    
---
relgenfun(2,1,1,I) 
relgenfun(2,1,2,I) 
\end{verbatim}
\end{procedure}

\begin{procedure}\label{second-weight-procedure} 
Computing the second Hamming weights of a toric code over a hypersimplex 
$\mathcal{C}_\mathcal{P}(d)$ on $T$. 
The input for this
procedure is a generating set of  $KV_{d}$ and the vanishing 
ideal $I$ of $T$. This procedure 
corresponds to Example~\ref{second-Hamming-weight-example}.  
\begin{verbatim}
q=3
K=ZZ/q;
R=K[t1,t2,t3,t4];
--This is the vanishing ideal of the affine torus T:
I=ideal(t1^(q-1)-1,t2^(q-1)-1,t3^(q-1)-1,t4^(q-1)-1)
-------
L=flatten entries basis(4,R)
--This is a K-basis for the linear space L:
M=toList (set apply(0..#L-1,x-> if (rsort flatten exponents L#x)#0 <= 1 
then L#x else 1)
-set{1})
-----
D=apply(toList (set(0..char ring I-1))^**(#M)-(set{0})^**(#M),toList);
K=apply(D,x->matrix {M}*vector deepSplice x)
--This is the linear space KV_{d}
T=apply(K,z->ideal(flatten entries z))
--number of zeros of each element in KV_{d}
H=apply(T,x -> degree(I+x))
--maximum number of zeros of a element in KV_{d}
max H
----number of zeros of a polynomial in KV_{d} with less zeros than max H
G=apply(T,x->if degree(I+x)<max H then degree(I+x) else 0)
F=set apply(T,x-> if  degree(I+x)==max G then x else 0)-set {0}
---
max G
--- The second Hamming weight
(q-1)^4-max G
\end{verbatim}
\end{procedure}

\begin{procedure}\label{sharp-inequ} 
Computing the relative generalized Hamming weights and relative footprint of an evaluation code $\mathcal{L}^{1}_{X}$ with respect to a subcode  $\mathcal{L}^{2}_{X}$ using Theorem \ref{relative-Hamming-weights}, and Theorem \ref{foorprint-relative-gen-weight}. The input for this procedure is two generating sets for the spaces $\mathcal{L}^{1}$, $\mathcal{L}^{2}$ and the vanishing ideal of $X$. This procedure 
corresponds to Example~\ref{sharp-bound-inequ}.  
\begin{verbatim}
q=5, K=GF(q,Variable=>a), S=K[t1,t2]
--This is the vanishing ideal of the affine torus T:
I=ideal(t1^(q-1)-1,t2^(q-1)-1)
--This is the quotient ring S/I:
M=coker gens gb I
r=1--we are computing the r-th relative generalized Hamming weight 
--Grobner basis
G=gb I
--This is a K-basis for the linear space L:
Basis={1,t1^3,t1*t2^2,t2^3,t1*t2,t1^2}
subBasis={t1*t2^2,t1*t2}
F=apply(toList((set(0,a,a^2,a^3,a^4))^**(#Basis)-(set{0})^**(#Basis))/deepSplice
    ,toList);
E=apply(F,x->matrix{Basis}*vector x);
R=apply(E,z-> flatten entries z);
L=apply(toList((set(0,a,a^2,a^3,a^4))^**(#subBasis)-(set{0})^**(#subBasis))
    /deepSplice,toList);
EE=apply(apply(L,x->matrix{subBasis}*vector x),z->flatten entries z);
MM=apply(R,x->if member(x,EE) then 0 else x)-set {0};
U=apply(subsets(apply(apply(MM,n->n#0),m->(leadCoefficient(m))^(-1)*m),r),ideal)
S=degree M- max apply(U,x->if #(set flatten entries leadTerm gens x)==r 
    then degree(I+x) else 0)
--This is the initial ideal of I:
init=ideal(leadTerm gens gb I)
er=(x)-> degree ideal(init,x)
--This gives the r-th footprint:
UU=toList apply(U,x->leadTerm gens x);
fpr=degree M - max apply(apply(apply(subsets(UU,r),toSequence),ideal),er) 
\end{verbatim}
\end{procedure}


\bibliographystyle{plain}

\end{document}